\documentclass[letterpaper, 10 pt, conference]{ieeeconf}  

\IEEEoverridecommandlockouts                              

\usepackage[tbtags]{amsmath}
\usepackage{amssymb} 
\usepackage{amsthm}
\usepackage{mathrsfs}
\usepackage{xifthen}
\usepackage{xstring} 
\usepackage{xspace}
\usepackage{caption}
\usepackage{subcaption} 
\usepackage{textcmds} 
\usepackage{mathtools}  
\usepackage{algorithm} 
\usepackage{algpseudocode} 
\usepackage{tikz}
\usepackage{pgfplots}
\usepackage[short]{optidef}
\usepackage{adjustbox}

\usepackage[unicode=true, colorlinks=false, breaklinks=true]{hyperref}
    
    \newcommand*{\ie}{i.e.,\@\xspace}
    \newcommand*{\eg}{e.g.,\@\xspace}
    \newcommand*{\cf}{c.f.\@\xspace}

    \newcommand{\figref}[1]{\hyperref[#1]{Fig.\ \ref*{#1}}}
    \newcommand{\tabref}[1]{\hyperref[#1]{Tab.\ \ref*{#1}}}
    \newcommand{\secref}[1]{\hyperref[#1]{Sec.\ \ref*{#1}}}
    \newcommand{\probref}[1]{\hyperref[#1]{Problem\ \ref*{#1}}}
    \newcommand{\theoref}[1]{\hyperref[#1]{Th.\ \ref*{#1}}}
    \newcommand{\lemmaref}[1]{\hyperref[#1]{Lemma\ \ref*{#1}}}
    \newcommand{\assref}[1]{\hyperref[#1]{Assumption\ \ref*{#1}}}
    \newcommand{\defref}[1]{\hyperref[#1]{Def.\ \ref*{#1}}}
    \newcommand{\propref}[1]{\hyperref[#1]{Proposition\ \ref*{#1}}}
    \newcommand{\reqref}[1]{\hyperref[#1]{Requirements\ \ref*{#1}}}
    \renewcommand{\algref}[1]{\hyperref[#1]{Alg.\ \ref*{#1}}}
    \newcommand{\remref}[1]{\hyperref[#1]{Remark\ \ref*{#1}}}
    
    \newcommand*{\fun}[1]{\ensuremath{\mathantt{#1}}}
    
    \newcommand*{\set}[1]{\ensuremath{\mathcal{#1}}}

    \newcommand*{\R}{\mathbb{R}}

    \DeclareMathOperator{\defi}{\coloneqq}
    \DeclareMathOperator{\rotdefi}{\eqqcolon}

	\newcommand*{\tn}[1]{\textnormal{#1}}
	
	\newcommand*{\mrm}[1]{\ensuremath{\mathrm{#1}}}
	
    \newcommand*{\grass}[2]{\ensuremath{\mathcal{GR}(#1,#2)}}
    
    \newcommand*{\stiefel}[2]{\ensuremath{\mathcal{ST}(#1,#2)}}
    
    \newcommand*{\og}[1]{\ensuremath{\mrm{O}({#1})}}
    
    \newcommand*{\eqclass}[1]{\ensuremath{ [#1] }}

    \DeclareMathOperator{\conv}{conv}
    
    \DeclareMathOperator{\Span}{span}
    
	\newcommand*{\valunit}[2]{\ensuremath{#1\,\mathrm{#2}}}

    \newcommand*{\initstate}[0]{\ensuremath{ x_\mrm{s} }}

    \usetikzlibrary{arrows,positioning,calc}
    \tikzset{%
        block/.style    = {draw, thick, rectangle, minimum height = 2em,
    minimum width = 3em},
        blockSmall/.style    = {draw, thick, rectangle, minimum height = 1.25em,
    minimum width = 1.75em},
        sum/.style      = {draw, circle, node distance = 2cm}, 
    }
    
    \definecolor{CCPSturquoise1}{RGB}{49,133,156}
    \definecolor{CCPSturquoise2}{RGB}{219,238,244}
    
    \definecolor{CCPSblue1}{RGB}{37,64,97}
    \definecolor{CCPSblue2}{RGB}{56,96,146}
    \definecolor{CCPSblue3}{RGB}{220,230,242}
    
    \definecolor{CCPSgreen1}{RGB}{67,80,58}
    \definecolor{CCPSgreen2}{RGB}{170,171,128}
    \definecolor{CCPSgreen3}{RGB}{215,228,189}
    
    \definecolor{CCPSred1}{RGB}{99,37,35}
    \definecolor{CCPSred2}{RGB}{149,55,53}
    \definecolor{CCPSred3}{RGB}{242,220,219}
    \definecolor{CCPSredBound}{rgb}{0.75,0,0}
    
    \definecolor{CCPSgray1}{RGB}{127,127,127}
    \definecolor{CCPSgray2}{RGB}{191,191,191}
    \definecolor{CCPSgray23}{RGB}{217,217,217}
    \definecolor{CCPSgray3}{RGB}{242,242,242}
    \definecolor{MPItext}{RGB}{56,60,60}
    
\DeclareMathAlphabet{\mathantt}{OT1}{antt}{li}{it}
\DeclareMathAlphabet{\mathpzc}{OT1}{pzc}{m}{it}

\makeatletter
\newcommand{\StatexIndent}[1][3]{%
	\setlength\@tempdima{\algorithmicindent}%
	\Statex\hskip\dimexpr#1\@tempdima\relax}
\makeatother

\title{\LARGE \bf
Geometric Data-Driven Dimensionality Reduction in MPC with Guarantees
}

\author{
    \authorblockN{
        Roland Schurig$^{1}$,
        Andreas Himmel$^{1}$,
        Rolf Findeisen$^{1}$%
    }
    \thanks{$^{1}$
        Technical University of Darmstadt, Control and Cyber-Physical Systems Laboratory, 64289 Darmstadt, Germany,  \{roland.schurig, andreas.himmel, rolf.findeisen\}@iat.tu-darmstadt.de
    }%
    \thanks{The authors acknowledge support within the frame of the LOEWE initiative emergenCITY.
    }%
}

\newtheorem{theorem}{Theorem}
\newtheorem{lemma}{Lemma}
\newtheorem{corollary}{Corollary}

\theoremstyle{definition}
\newtheorem{definition}{Definition}
\newtheorem{problem}{Problem}
\newtheorem{assumption}{Assumption}

\newtheorem{proposition}{Proposition}

\theoremstyle{remark}
\newtheorem{remark}{Remark}
\newtheorem{example}{Example}

\begin{document}

\maketitle
\thispagestyle{empty}
\pagestyle{empty}

\begin{abstract}

We address the challenge of dimension reduction in the discrete-time optimal control problem which is solved repeatedly online within the framework of model predictive control.
Our study demonstrates that a reduced-order approach, aimed at identifying a suboptimal solution within a low-dimensional subspace, retains the stability and recursive feasibility characteristics of the original problem. 
We present a necessary and sufficient condition for ensuring initial feasibility, which is seamlessly integrated into the subspace design process. 
Additionally, we employ techniques from optimization on Riemannian manifolds to develop a subspace that efficiently represents a collection of pre-specified high-dimensional data points, all while adhering to the initial admissibility constraint.
\end{abstract}

\section{INTRODUCTION}

In model predictive control (MPC) \cite{Findeisen2002, Gruene2017,Borrelli2017}, the control input is computed as the solution of an optimal control problem.
The optimization must be performed online and in real time at every sampling instance of the control system. 
This limits its applicability for systems with small sampling times and/or limited hardware capacities.

Existing work to counteract this problem aims at alleviating the computational burden of solving the optimal control problem.
A prominent idea is to approximate the problem in a lower-dimensional subspace.
The design of the subspace can range from conceptually simple approaches like \emph{move blocking} \cite{Cagienard2007, Shekhar2015} to more elaborated algorithms that involve \emph{singular value decomposition} \cite{Pan2023, Bemporad2020, Rojas2004}.

Many of these ideas share the paradigm that they are \emph{data-driven}, \ie the high-dimensional problem has to be sampled and the subspace is designed based on this data to capture its the most important patterns. 
Furthermore, some results for the stability of such schemes exist, \eg \cite{Cagienard2007, Pan2023}.

The goal of this contribution is to recall, collect and unify existing results for subspace-based dimensionality reduction in MPC.
We carefully review the dimensionality reduction problem in the linear MPC context and present a general scheme, based on \cite{Pan2023, Bemporad2020, Goebel2017}, which enjoys stability and feasibility guarantees.
We provide self-contained proofs within our setting and notation, and add an upper bound for the closed-loop cost.
Then, we relate the scheme to the geometric framework for dimensionality reduction introduced in \cite{Schurig2023}.

We put special focus on the aspect of \emph{initial} feasibility. 
We build upon that and add necessary and sufficient conditions for initial feasibility in a given initial set and translate them into our geometric framework.
While \cite{Schurig2023} mainly focused on the geometry of the subspace design, we also relate it lengthily to the MPC formulation.

We firmly believe that the geometric framework has numerous advantages, since it exploits the geometry of the subspace approximation and allows for a conceptually conclusive problem formulation. 
The growing number of methods and algorithms for solving (constrained) optimization problems over Riemannian manifolds pave the way for our formulation \cite{Absil2008, Boumal2023, Liu2020}.

\paragraph*{Notation}
To denote nonnegative real numbers, we use $\R_{\geq 0} \defi \{ x \in \R \mid x \geq 0 \}$. 
The positive integers are $\mathbb{N}$, and $\mathbb{N}_0 \defi \mathbb{N} \cup \{ 0 \}$. 
We use the standard notion of class $\mathscr{K}$ and $\mathscr{K}_\infty$ comparison functions \cite[Def.~2.13]{Gruene2017}.

\section{MPC FORMULATION}
\label{sec:mpc}

Consider the problem of steering the discrete-time linear time-invariant system \\[-3.5ex]
\begin{align}
	x(t+1) &= A x(t) + B u(t), &
	x(0) &= \initstate,
	\label{eq:plant}
\end{align} \\[-3.5ex]
to the origin, with the pair $(A,B)$ stabilizable, while obeying the constraints $x(t) \in \mathbb{X} \subseteq \R^n$, $u(t) \in \mathbb{U} \subseteq \R^m$, $\forall t \in \mathbb{N}_0$, where $\mathbb{X}$ and $\mathbb{U}$ are compact polyhedral sets containing the origin in their interior.

A suitable method to tackle this problem is MPC \cite{Findeisen2002, Gruene2017,Borrelli2017}. 
It relies on repeatedly solving an open-loop optimal control problem over a finite prediction horizon $N \in \mathbb{N}$.

In MPC, a \emph{stage cost} $\ell \colon \R^n \times \R^m \to \R_{\geq 0}, \, (x,u) \mapsto \ell(x,u)$ that penalizes the distance of the pair $(x,u)$ from the origin is considered.
Then, at each sampling time, MPC predicts $N$ time steps into the future and chooses a control sequence that minimizes the summed stage cost over the horizon. 
The first part of the control sequence is applied to the plant, and the procedure is repeated at the next time step. 
We assume that the stage cost $\ell$ is quadratic with $\ell(x,u) \defi x^\top Q x + u^\top R u$. 
The weighting matrices satisfy $Q \succ 0$ and $R \succ 0$. \\[-3ex] 
\begin{definition}
	We identify the set of input sequences over the horizon $N$ with $\R^d$, where $d \defi Nm$.
	Therefore, we call $z \in \R^d$ an \emph{input sequence} of length $N$.
	Given an input sequence $z = (z_0,\dots,z_{N-1}) \in \R^d$, we use $k \mapsto \fun{pr}_k(z) \defi z_k \in \R^m$
	for $k \in \{ 0, \dots, N-1 \}$ to select its $k^\mrm{th}$ component. 
\end{definition}
It is well-known that numerical problems and ill-conditioning in representing open-loop trajectories over a (long) horizon $N$ arise if the plant \eqref{eq:plant} is open-loop unstable.
Often, this problem is avoided by \emph{pre-stabilizing} the system with a linear controller 
$x \mapsto \kappa(x) \defi K x$
for some stabilizing feedback gain $K \in \R^{m \times n}$ \cite{Cagienard2007, Pan2023, Rojas2004}.

\begin{definition}
    Given an initial value $x_0 \in \R^n$ and an input sequence $z \in \R^d$, we define the pre-stabilized \emph{trajectory} $\fun{x}_z(\cdot, x_0)$ iteratively via $\fun{x}_z(0, x_0) \defi x_0$ and
	\begin{align*}
		\fun{x}_z(k+1, x_0) &\defi A \fun{x}_z(k, x_0) + B \left( \kappa(\fun{x}_z(k, x_0)) + \fun{pr}_k(z) \right)
	\end{align*}
	for $k \in \{ 0, \dots, N-1 \}$.
	We refer to
	\begin{align*}
		\fun{u}_z(k,x_0) &\defi \kappa(\fun{x}_z(k, x_0)) + \fun{pr}_k(z), &
		k &\in \{0,\dots,N-1\}
	\end{align*}
	as the corresponding \emph{control sequence} that produces the trajectory $\fun{x}_z(\cdot,x_0)$ according to the system dynamics \eqref{eq:plant}.
\end{definition}

We further introduce the following standard terminal ingredients that are frequently used in MPC to establish formal stability guarantees \cite{Gruene2017,Borrelli2017}.
\begin{assumption}
    \label{ass:terminal_ingredients}
	There exists a closed polyhedral terminal set $\mathbb{X}_\mrm{f} \subseteq \mathbb{X}$ containing the origin in its interior, a quadratic terminal penalty $V_\mrm{f} \colon \mathbb{X}_\mrm{f} \to \R_{\geq 0}$ and a terminal controller $\kappa_\mrm{f} \colon \mathbb{X}_\mrm{f} \to \mathbb{U}$ that satisfy
	$
		A x + B \kappa_\mrm{f}(x) \in \mathbb{X}_\mrm{f}
	$
	and
	$
		V_\mrm{f}(x) - V_\mrm{f} ( A x + B \kappa_\mrm{f}(x) ) \geq \ell(x,\kappa_\mrm{f}(x))
	$
	for all $x \in \mathbb{X}_\mrm{f}$.
\end{assumption}
\begin{definition}
	For a state $x \in \mathbb{X}$, we call the input sequence $z \in \R^d$ \emph{admissible} if
	\begin{equation*}
		\forall k \in \{0,\dots,N-1\} \colon (\fun{x}_z(k,x), \fun{u}_z(k,x) ) \in \mathbb{X} \times \mathbb{U} 
	\end{equation*}
	and $\fun{x}_z(N,x) \in \mathbb{X}_\mrm{f}$ holds \cite{Gruene2017}.
	The set of admissible input sequences for $x$ is denoted by $\mathbb{U}^N(x)$.
	
	Additionally, we introduce the \emph{feasible set}
	$
		\set{X}_N = \{ x \in \mathbb{X} \mid \mathbb{U}^N(x) \neq \emptyset \}
	$
	as the set of states for which we can find at least one admissible input sequence.
\end{definition}
Finally, for $x \in \R^n$ and $z \in \R^d$, we can introduce the cost function over the horizon $N$ as
$
    J_N(x,z) \defi V_\mrm{f}(\fun{x}_z(N,x)) + \sum_{k=0}^{N-1} \ell(\fun{x}_z(k,x),\fun{u}_z(k,x))
$.
Then, we define the \emph{constrained finite time optimal control problem}
\begin{align}
    \mathscr{P}_N(x) \colon &\min_{z \in \R^d} \; J_N(x,z) &
    &\mrm{s.t.} \; \, z \in \mathbb{U}^N(x)
    \label{eq:CFTOC} 
\end{align}
on which the standard MPC scheme is based.
We introduce the following objects to describe $\mathscr{P}_N(x)$:
	The \emph{optimal value function} $V_N \colon \set{X}_N \to \R_{\geq 0}$, which maps any $x \in \set{X}_N$ to the corresponding optimal value of \eqref{eq:CFTOC}, \ie $V_N(x) \leq J_N(x,z)$ holds for all $z \in \mathbb{U}^N(x)$.
	The \emph{optimizer mapping} $\mu_N \colon \set{X}_N \to \R^d$, which maps any $x \in \set{X}_N$ to the corresponding optimal input sequence of \eqref{eq:CFTOC}, \ie $V_N(x) = J_N(x,\mu_N(x))$.

The MPC controller asymptotically stabilizes \eqref{eq:plant} on the feasible set $\set{X}_N$ under the assumptions made. 
For more details, we refer to, \eg \cite[Th.~5.13]{Gruene2017}. 
Because the region of attraction is $\set{X}_N$, we introduce an additional assumption on the initial state of the plant.
\begin{assumption}
	\label{ass:initial_value}
	The initial value $\initstate$ of \eqref{eq:plant} is contained in the compact set $\mathbb{X}_0 = \conv(X_0) \subseteq \set{X}_N$, where ${X}_0 \defi \{ \bar{x}_1, \dots, \bar{x}_s \}$, $\bar{x}_j \in \set{X}_N$ and $\conv$ denotes the convex hull of a set.
\end{assumption}

\begin{remark}
    Problem \eqref{eq:CFTOC} is a \emph{multiparametric quadratic program} with parameter $x$ and optimization variable $z$ \cite[Sec.~6.3]{Borrelli2017}.
    The admissible sets $\mathbb{U}^N(x)$ are bounded polyhedra.
\end{remark}

\section{DIMENSIONALITY REDUCTION IN MPC}
\label{sec:reduced_mpc}

We aim to simplify optimization problem \eqref{eq:CFTOC}, which must be solved at every sampling instance of an MPC scheme.
A classical approach to diminish the computational burden in $\mathscr{P}_N(x)$ is to reduce the number of (scalar) decision variables from $d$ to $r$, with $r \ll d$; see, \eg \cite{Cagienard2007,Shekhar2015,Pan2023, Bemporad2020,Schurig2023}. 

One of the most established techniques and a cornerstone in statistics for dimensionality reduction and identifying patterns in data is \emph{principal component analysis} (PCA); see, \eg \cite{Boumal2023,Brunton2019}. 
PCA offers a systematic way to uncover the $r$ \emph{dominant directions} of variation of the data \cite{Boumal2023}. 
The underlying assumption is that the data is located on or near a $r$-dimensional affine subspace in $\R^d$. 
By transforming to this new coordinate system, a low-dimensional approximation of the high-dimensional data is available \cite{Brunton2019}.

Motivated by PCA, we aim for a {data-driven} design of suitable subspaces in the context of MPC.
We consider a data set of the form
\begin{equation}
	\set{D} = \{ (x_i, z_i) \}_{i=1}^L \subset \set{X}_N \times \R^d,
	\label{eq:systemData}
\end{equation}
where $z_i \defi \mu_N(x_i)$, \ie our data consists of pairs of states \emph{and} corresponding optimal input sequences.
We seek to find an affine subspace -- represented by a \emph{basis }and an \emph{offset} -- that best captures the dominant directions of variation in $\set{D}$. 

The \emph{compact Stiefel manifold}
$$
	\stiefel{r}{d} \defi \{ U \in \R^{d \times r} \mid U^\top U = I_r \} 
$$
offers a convenient way to represent a basis, since the columns of $U \in \stiefel{r}{d}$ form an orthonormal basis of an $r$-dimensional subspace of $\R^d$.
We refer to \cite{Edelman1998} for an introduction to the Stiefel manifold.

In order to capture the state information in $\mathcal{D}$, we allow for a \emph{state-dependent} offset
$
	\sigma \colon \set{X}_N \to \R^d 
$,
\ie we seek to identify a basis $U \in \stiefel{r}{d}$ whose column space best approximates $z_i - \sigma(x_i)$.
We restrict the class of possible offsets to \\[-2ex]
$$
    \set{F}_N \defi \{ \sigma \colon \set{X}_N \to \R^d \mid \sigma \colon x \mapsto \Gamma x + \xi \}.
$$ \\[-3ex] 
This has desirable properties when it comes to guaranteeing feasibility in the MPC context and combines ideas from previous works \cite{Bemporad2020, Goebel2017}.
For a discussion on how to obtain $(U,\sigma) \in \stiefel{r}{d} \times \set{F}_N$, we refer to \secref{sec:data_driven_design}.

Thus, for a given pair $(U,\sigma)$, an intuitive, low-dimensional reformulation of \eqref{eq:CFTOC} is to replace the optimization variable $z$ with $U \alpha + \sigma(x)$.
Note that this reduces the number of scalar decision variables from $d$ to $r$.
However, this intuitive reformulation does not yet fulfill the requirements we have if the optimization problem is to be solved in closed-loop in a receding horizon fashion. 
In particular, we must put a special focus on the issue of \emph{feasibility}.
More precisely, we shall discuss if and when for all states $x \in \set{X}_N$ there exists an $\alpha_x \in \R^r$ such that $U \alpha_x + \sigma(x) \in \mathbb{U}^N(x)$ holds.

In general, the answer to this question depends on how the pair $(U,\sigma)$ is chosen. 
However, a common approach in reduced MPC to circumvent this problem is to hand an admissible guess $\tilde{z} \in \mathbb{U}^N(x)$ -- however obtained -- to the solver as a \qq{fall-back} option \cite{Bemporad2020}.

Therefore, given $x \in \set{X}_N$ and $\tilde{z} \in \R^d$, we introduce the \emph{reduced-order} MPC problem
\begin{mini}
	{(\alpha,\tau)}{ J_N(x, U \alpha + \tau \sigma(x) + (1-\tau) \tilde{z} ) }
	{\label{eq:CFTOC_reduced}}{\widetilde{\mathscr{P}}_N(x,\tilde{z}) \colon}
	\addConstraint{U \alpha + \tau \sigma(x) + (1-\tau) \tilde{z}}{\in \mathbb{U}^N(x)}
\end{mini}
to replace $\mathscr{P}_N(x)$, with $\tau \in \R$. 
Note that the optimization problem accepts an additional input $\tilde{z}$ and we have added an extra degree of freedom in $\tau$.

\begin{remark}
	If $\tilde{z} \in \mathbb{U}^N(x)$, then $\widetilde{\mathscr{P}}_N(x,\tilde{z})$ is guaranteed to have the (suboptimal) solution $\tilde{\alpha}=0$ and $\tilde{\tau} = 0$.
	Thus, we have introduced $\tau$ to ensure that in this case the solver can always return an admissible input sequence by scaling the influence of the basis and the offset to zero.
\end{remark}

Similar to the full-order case, we use $\widetilde{V}_N(x,\tilde{z})$ to denote the resulting optimal value, which depends on the choice of $(U,\sigma)$. 
However, to not overload notation and maintain readability, we suppress this dependence here.

Suppose that the optimizer of $\widetilde{\mathscr{P}}_N(x,\tilde{z})$ is $(\alpha^\star, \tau^\star)$. 
Then, the corresponding optimal input sequence is given by 
$
	\tilde{\mu}_N(x,\tilde{z}) \defi U \alpha^\star + \tau^\star \sigma(x) + (1-\tau^\star) \tilde{z} ,
$
\ie $\widetilde{V}_N(x,\tilde{z}) = J_N(x,\tilde{\mu}_N(x,\tilde{z}))$.

We now shift our focus to the question of how admissible guesses $\tilde{z} \in \mathbb{U}^N(x)$ can be obtained. 
In MPC with terminal conditions, a common approach to construct an admissible solution at time $t \in \mathbb{N}$ \emph{in closed-loop} is to consider the admissible optimal solution from the \emph{last} time step $t-1$, \cf \cite{Cagienard2007,Shekhar2015,Pan2023,Bemporad2020}. 
The procedure is based on the terminal ingredients in \assref{ass:terminal_ingredients} and relies on standard arguments.
\begin{proposition}
    \label{prop:admissible_shift}
	Given a state $x \in \set{X}_N$ and an input sequence $z \in \mathbb{U}^N(x)$, we create an input sequence $\fun{s}_\mrm{f}(x,z) \in \mathbb{U}^N(\fun{x}_z(1,x))$. 
    If $\fun{x}_z(1,x) = 0$, then we set $\fun{s}_\mrm{f}(x,z) \defi 0$.
    Otherwise, we define
	\begin{equation*}
		\fun{pr}_k(\fun{s}_\mrm{f}(x,z)) \defi \begin{cases*}
			\fun{pr}_{k+1}(z), & $k \in \{0,\dots,N-2\}$ \\
			\kappa_\mrm{f}( x_{N} ) - \kappa( x_{N} ), & $k = N-1$
		\end{cases*}
	\end{equation*}
    with $x_{N} = \fun{x}_z(N,x)$.
    We call $\fun{s}_\mrm{f}(x,z)$ the \emph{admissibly shifted input sequence} for $x$ and $z$. 
\end{proposition}
\begin{proof}
    We show that the inclusion $\fun{s}_\mrm{f}(x,z) \in \mathbb{U}^N(\fun{x}_z(1,x))$ is true. 
    For $\fun{x}_z(1,x) = 0$, it trivially holds. 
    For $\fun{x}_z(1,x) \neq 0$, the inclusion is implied by construction of $\fun{s}_\mrm{f}$ that exploits the properties in \assref{ass:terminal_ingredients}; see, \eg \cite[Lemma~5.10]{Gruene2017}.
\end{proof}
Therefore, if $x \in \set{X}_N$ and $z^\star \in  \mathbb{U}^N(x)$ denote the state and the optimal input sequence at time $t-1 \in \mathbb{N}$, we can hand $\tilde{z}^+ = \fun{s}_\mrm{f}(x,z^\star)$ to \eqref{eq:CFTOC_reduced} as the admissible guess for the state $x^+ = \fun{x}_{z^\star}(1,x)$ at time $t$. 

With this reasoning, however, we \emph{cannot} rely on a optimal solution from the last time step at time $t=0$.
From \assref{ass:initial_value} we know that the initial state $\initstate$ at time $t=0$ is in $\mathbb{X}_0$.
Consequently, we must establish the existence of admissible input sequences in \eqref{eq:CFTOC_reduced} for all $\initstate \in \mathbb{X}_0$ by just relying on $U$ and $\sigma$.

We shift the discussion on how this can be incorporated into the design of $U$ and $\sigma$ to \secref{sec:data_driven_design}. 
For now, we content ourselves with the following definition.
\begin{definition}
	\label{def:initial_feasibility}
	We call the pair $(U,\sigma)$ given by a basis $U \in \stiefel{r}{d}$ and an offset $\sigma \in \set{F}_N$ \emph{initially admissible} (IA) if 
	\begin{equation*}
		\forall \initstate \in \mathbb{X}_0 \colon \exists {\alpha} \in \R^r \colon  U {\alpha} + \sigma(\initstate) \in \mathbb{U}^N(\initstate) .
	\end{equation*}
\end{definition}

Throughout the remainder of the section, we assume that the pair $(U,\sigma)$ is initially admissible. 
This notion allows the introduction of \algref{alg:reducedRHC} as our reduced-order MPC scheme. 
It is based on the reduced problem \eqref{eq:CFTOC_reduced} and the admissible shift $\fun{s}_\mrm{f}$. 
In the following, we analyze the closed-loop system resulting from controlling \eqref{eq:plant} with \algref{alg:reducedRHC}.

\begin{algorithm}[t]
	\caption{Reduced-order MPC scheme.}\label{alg:reducedRHC}
	\begin{algorithmic}[1]
		\Require Initially admissible pair $(U,\sigma) \in \stiefel{r}{d} \times \set{F}_N$.
		\State Initialize $\tilde{z} \gets 0$.
		\Repeat{ At each sampling time $t \in \mathbb{N}_0$ ...}
		\State Measure the system state $x(t) \in \mathbb{X}$ and set $x \defi x(t)$.
		\State Solve $\widetilde{\mathscr{P}}_N(x,\tilde{z})$ in \eqref{eq:CFTOC_reduced} with $(U,\sigma)$ and denote the \StatexIndent[1] 
		resulting input sequence by $z^\star \defi \tilde{\mu}_N(x,\tilde{z})$. 
        \State Update $\tilde{z} \gets \fun{s}_\mrm{f}(x,z^\star)$. 
		\State Define the reduced-order MPC feedback \StatexIndent[1] $\tilde{\kappa}_N(x,\tilde{z}) \defi \fun{u}_{z^\star}(0,x)$ and apply it to the system.
		\Until{system state has converged to the origin;}
	\end{algorithmic}
\end{algorithm}

The first important point to recognize is that the computation of the control value at time $t$ does \emph{not} only depend on the current state $x$ of the plant, but also on the admissible guess $\tilde{z}$ that is updated throughout the algorithm based on the solution of \eqref{eq:CFTOC_reduced}. 
This is the case because for $\tilde{z}_1, \tilde{z}_2 \in \mathbb{U}^N(x)$ with $\tilde{z}_1 \neq \tilde{z}_2$, we have $\tilde{\mu}_N(x,\tilde{z}_1) \neq \tilde{\mu}_N(x,\tilde{z}_2)$ in general.
Therefore, the state information $x$ of the plant \eqref{eq:plant} is \emph{not} sufficient to describe what its next state $x^+ = A x + B \tilde{\kappa}_N(x,\tilde{z})$ in closed-loop will be; the information $\tilde{z}$ on the current admissible guess is required as well.

For this reason, we work with an extended state, which consists of the plant's state \emph{and} the admissible guess. 
For $\initstate \in \mathbb{X}_0$, we introduce $\mathpzc{x}(0,\initstate) \defi \initstate$, $\mathpzc{z}(0,\initstate) \defi 0$ and
\begin{equation*}
\begin{aligned}
    \mathpzc{x}(t+1,\initstate) &\defi A \mathpzc{x}(t,\initstate) 
    + B \tilde{\kappa}_N(\mathpzc{x}(z,\initstate),\mathpzc{z}(t,\initstate)), \\
    \mathpzc{z}(t+1,\initstate) &\defi \fun{s}_\mrm{f}(\mathpzc{x}(t,\initstate),\tilde{\mu}_N(\mathpzc{x}(t,\initstate),\mathpzc{z}(t,\initstate))) .
\end{aligned}
\end{equation*} 
Here, $\mathpzc{x}(t,\initstate)$ denotes the state of the plant in closed-loop at time $t \in \mathbb{N}_0$ when the initial state at time $t=0$ was $\initstate$, while $\mathpzc{z}(t,z_\mrm{s})$ denotes the admissible guess at the same time. 

We define the state space of the closed-loop dynamics as $M \defi M_1 \cup M_2 \subseteq \set{X}_N \times \R^d$ with
$
    M_1 \defi \{ (\initstate,0) \mid \initstate \in \mathbb{X}_0 \}
$ and
$
    M_2 \defi \{ (x,\tilde{z}) \mid x \in \set{X}_N \setminus \{ 0 \} , \, \tilde{z} \in \mathbb{U}^N(x) \} \cup \{ (0, 0) \}
$. 
\begin{remark}
    This somewhat overelaborate definition serves two purposes. 
    On the one hand, it explicitly gives prominence to the nature of the initial extended state in \algref{alg:reducedRHC}.
    On the other hand, it highlights the fact that by construction of $\fun{s}_\mrm{f}$, once the plant state $x$ reaches the equilibrium, it stays there. (We show that in our next result.)
    We choose this definition to stress that (i) the solution of \eqref{eq:CFTOC_reduced} is special at time $t=0$ in the sense that feasibility solely originates from $(U,\sigma)$ and (ii) the entire scheme only works under the initial admissibility assumption on $(U,\sigma)$. 
    For subsequent times, feasibility of \eqref{eq:CFTOC_reduced} will stem from \propref{prop:admissible_shift}.
\end{remark}

We introduce the map
$
    g \colon M \to \R^n \times \R^d, \, p = (x,\tilde{z}) \mapsto g(p) = (x^+,\tilde{z}^+) 
$,
with $x^+ \defi A x + B \tilde{\kappa}_N(x,\tilde{z})$ and $\tilde{z}^+ \defi \fun{s}_\mrm{f}(x,\tilde{\mu}_N(x,\tilde{z}))$, and establish some basic properties.

\begin{lemma}
    The value $g(p)$ exists for each $p \in M$ and the inclusion $g(p) \in M_2$ holds for each $p \in M$. 
    \label{lemma:closed_loop_existence}
\end{lemma}
\begin{proof}
    Assume first that $(x,\tilde{z}) \in M_2$. 
    Then, $\widetilde{\mathscr{P}}_N(x,\tilde{z})$ is feasible by construction, so $\tilde{\mu}_N(x,\tilde{z})$ exists. 
    If $(x,\tilde{z}) \in M_1$, then $\tilde{\mu}_N(x,\tilde{z})$ is guaranteed to exist under the IA assumption on $(U,\sigma)$.
    Thus, we have the existence of $\tilde{\mu}_N(x,\tilde{z}) \in \mathbb{U}^N(x)$ and hence also of $\tilde{\kappa}_N(x,\tilde{z})$ for all $(x,\tilde{z}) \in M$.
    Also, for $(x,\tilde{z}) \in M$, we have $x \in \set{X}_N$, so $\fun{s}_\mrm{f}(x,\tilde{\mu}_N(x,\tilde{z}))$ is well-defined. 
    This establishes the existence of $g(p)$ for all $p \in M$.

    For the second part, we have to show that $g(p)$ is never of the form $(0,\tilde{z})$ with $\tilde{z} \neq 0$, since $(0,\tilde{z})$ implies $\tilde{z} = 0$ by construction of $M_2$.
    We first consider $p = (0,0)$. 
    Then, we have $\tilde{\mu}_N(0,0) = 0$, which implies $g(p) = (0,0) \in M_2$.
    Next, we consider $p = (x,\tilde{z})$ with $A x + B \tilde{\kappa}_N(x,\tilde{z}) = 0$. 
    Then, $\fun{s}_\mrm{f}(x,\tilde{\mu}_N(x,\tilde{z})) = 0$ by defnition of $\fun{s}_\mrm{f}$, which implies $g(p) = (0,0) \in M_2$.
    Finally, for all remaining $p = (x,\tilde{z}) \in M$, we have $g(p) = (x^+,\tilde{z}^+)$ with $x^+ \neq 0$.
    By \propref{prop:admissible_shift}, the inclusion $\tilde{z}^+ \in \mathbb{U}^N(x^+)$ holds, which confirms the claim $g(p) \in M_2$ for all $p \in M$. 
\end{proof}
This result justifies to introduce the difference equation
\begin{align}
    p(t+1) &= g(p(t)), &
	p(0) &= p_\mrm{s} \in M.
    \label{eq:extended_closed_loop}
\end{align} 
We denote the solution of \eqref{eq:extended_closed_loop} with initial state $p_\mrm{s} \in M$ by $\mathpzc{p}(t,p_\mrm{s})$. 
Note that for $p_\mrm{s} = (\initstate,0) \in M_1$, the equality $\mathpzc{p}(t,p_\mrm{s}) = (\mathpzc{x}(t,\initstate),\mathpzc{z}(t,\initstate)$ holds, \ie $\mathpzc{p}(\cdot,p_\mrm{s})$ represents the trajectory resulting from controlling \eqref{eq:plant} with \algref{alg:reducedRHC}.

Hence, we have established that \eqref{eq:extended_closed_loop} is well-defined and that it describes the closed-loop dynamics of our reduced MPC scheme. 
Next, we turn our attention to the closed-loop properties. 
A contribution of this work is to explicitly state an upper bound of the closed-loop cost of our reduced MPC scheme. 
For $p = (x,\tilde{z}) \in M$, we introduce $\tilde{\ell}(p) \defi \ell(x,\tilde{\kappa}_N(x,\tilde{z}))$ as a short-hand notation for the stage cost in closed-loop for the reduced MPC controller.
Using the closed-loop dynamics \eqref{eq:extended_closed_loop} with initial state $p_\mrm{s} = (\initstate,0) \in M_1$, we introduce
$
    \tilde{J}_\mrm{cl}(\initstate) \defi \sum_{t=0}^\infty \tilde{\ell} \left(
	\mathpzc{p}(t,p_\mrm{s}) \right)
$
as the \emph{closed-loop cost} of the reduced MPC controller for $\initstate \in \mathbb{X}_0$.
We first state an intermediate result, \cf \cite{Pan2023,Goebel2017}.
\begin{lemma}
    \label{lemma:lyap_decrease}
    The inequality $\tilde{V}_N(g(p)) \leq \tilde{V}_N(p) - \tilde{\ell}(p)$ holds for all $p = (x,\tilde{z}) \in M$. 
\end{lemma}
\begin{proof}
    We denote $z^\star = \tilde{\mu}_N(p)$, $u = \tilde{\kappa}_N(p)$, $x_N = \fun{x}_{z^\star}(N,x)$ and $g(p) = (x^+,\tilde{z}^+) \in M$.
    By construction, $x_N \in \mathbb{X}_\mrm{f}$, $\fun{u}_{\tilde{z}^+}(N-1,x^+) = \kappa_\mrm{f}(x_N)$ and $\fun{x}_{\tilde{z}^+}(k,x^+) = \fun{x}_{z^\star}(k+1,x)$ for $k \in \{0,\dots,N-1\}$.
	This implies $J_N(g(p)) = V_\mrm{f}(A x_N + B \kappa_\mrm{f}(x_N)) + \tilde{V}_N(p) - \ell(x,u) - V_\mrm{f}(x_N) + \ell(x_N,\kappa_\mrm{f}(x_N))$. 
	Thus, with \assref{ass:terminal_ingredients} and $\tilde{V}_N(g(p)) \leq J_N(g(p))$, we obtain $\tilde{V}_N(g(p)) \leq \tilde{V}_N(p) - \ell(x,u)$. 
\end{proof}
\begin{theorem}
    The closed-loop cost $\tilde{J}_\mrm{cl}(\initstate)$ is upper bounded by $\widetilde{V}_N(\initstate,0)$ for all $\initstate \in \mathbb{X}_0$.	
    \label{theo:closed_loop_upper_bound}
\end{theorem}
\begin{proof}
    We follow the proof of \cite[Th.~4.11]{Gruene2017}. 
    By \lemmaref{lemma:lyap_decrease}, the dynamic programming inequality
    $
    \tilde{V}_N(p) - \tilde{V}_N(g(p)) \geq \tilde{\ell}(p)
    $
    holds for all $p \in M$. 
    Then, under \lemmaref{lemma:closed_loop_existence}, considering a closed-loop trajectory $\mathpzc{p}(\cdot,p_\mrm{s})$ starting from $p_\mrm{s} = (\initstate,0) \in M_1$ and summing over $t$ gives for all $K \in \mathbb{N}$
    \begin{multline*}
        \sum_{t=0}^K \tilde{\ell} \left(
	       \mathpzc{p}(t,p_\mrm{s}) \right)
        \leq \sum_{t=0}^{K-1} \tilde{V}_N(\mathpzc{p}(t,p_\mrm{s})) - \tilde{V}_N(\mathpzc{p}(t+1,p_\mrm{s})) \\
        = \tilde{V}_N(\mathpzc{p}(0,p_\mrm{s})) - \tilde{V}_N(\mathpzc{p}(K,p_\mrm{s})) \leq \tilde{V}_N(p_\mrm{s}) = \tilde{V}_N(\initstate,0) .
    \end{multline*}
    Since $\ell$ is nonnegative, the term on the left is monotone increasing and bounded. 
    Therefore, for $K \to \infty$, it converges to $\tilde{J}_\mrm{cl}(\initstate)$. 
    The assertion follows because the right-hand side of the inequality is independent of $K$.
\end{proof}

Finally, we focus on stability. 
Similar results can be found in \cite{Pan2023,Goebel2017}. 
We show that the plant's state converges to the origin in closed-loop. 
To exclude that the state deviates too far from to origin before it eventually converges to it, we show that \algref{alg:reducedRHC} leads to a stable closed-loop system in the sense of \cite[Def.~10.24]{Gruene2017}. 
While this result does not establish asymptotic stability, it suffices within the context of this work. 
\begin{theorem}
    The solution $\mathpzc{p}(\cdot,p_\mrm{s}) = (\mathpzc{x}(\cdot,\initstate),\mathpzc{z}(\cdot,\initstate))$ of \eqref{eq:CFTOC_reduced} satisfies $\lim_{t \to \infty} \mathpzc{x}(t,\initstate) = 0$ for all $p_\mrm{s} = (\initstate,0) \in M_1$.
    In addition, consider any compact set $S \subseteq \mathbb{X}_0$ with $0 \notin S$. 
    There exists $\alpha_\mrm{S} \in \mathscr{K}$ such that $\| \mathpzc{x}(t,\initstate) \| \leq \alpha_\mrm{S}(\| \initstate \|)$ holds for all $\initstate \in S$ and $t \in \mathbb{N}_0$.
\end{theorem}
\begin{proof}
    We follow the proof of \cite[Proposition~10.25]{Gruene2017}. 
    First, recall from the proof of \theoref{theo:closed_loop_upper_bound} that 
    $
    \lim_{K \to \infty} \sum_{t=0}^K \tilde{\ell} \left(
	       \mathpzc{p}(t,p_\mrm{s}) \right) \leq \tilde{V}_N(p_\mrm{s}) 
    $
    holds. 
    Nonnegativity of $\tilde{\ell}$ then implies $\lim_{t \to \infty} \tilde{\ell}\left(
	       \mathpzc{p}(t,p_\mrm{s}) \right) = 0$.
    In general, we have $\tilde{\ell}(p) = x^\top Q x + \tilde{\kappa}_N(p)^\top R \tilde{\kappa}_N(p) \geq x^\top Q x \geq \lambda_\mrm{min}(Q) \| x \|^2$ for any $p = (x,\tilde{z}) \in M$, where we have used $Q \succ 0$.
    This implies $\mathpzc{x}(t,\initstate) \to 0$.     

    For the second part, we first show that $\tilde{V}_N(\cdot,0)$ is upper-bounded by a $\mathscr{K}_\infty$ function on $S$, for which we apply ideas taken from the proof of \cite[Proposition~5.7]{Gruene2017}. 
    The value $\tilde{V}_N(\initstate,0)$ exists for each $\initstate \in S$ by the initial admissibility assumption. 
    The positive value $l_\mrm{min} \defi \min_{\initstate \in S} \| \initstate \|$ exists by compactness of $S$. 
    We define $N_l \defi B_{l}(0) \cap S$, where $B_{l}(0)$ denotes the closed ball around $0$ with radius $l$. 
    Note that $N_l$ is compact, which implies that $N_l \times \mathbb{U}^N$ is compact. 
    Thus, for $l \geq l_\mrm{min}$, we can compute 
    \begin{multline*}
        \hat{\alpha}_2(l) \defi \max \{V_\mrm{f}(x_N) + \sum_{k=0}^{N-1} \ell(x_k,u_k) \mid \\
        x_0 \in N_l, \, (u_0, \dots, u_{N-1}) \in \mathbb{U}^N, \, {x}_{k+1} = A x_k+ B u_k
        \}
    \end{multline*}
    which implies $\tilde{V}_N(\initstate,0) \leq \hat{\alpha}_2(\| \initstate \|)$ for all $\initstate \in S$. 
    The resulting function $\hat{\alpha}_2$ is continuous and monotone increasing in $l$.
    Next, we pick an $\alpha_2^\prime \in \mathscr{K}_\infty$ with $\alpha_2^\prime(l_\mrm{min}) = \hat{\alpha}_2(l_\mrm{min})$ and any $\bar{\alpha}_2 \in \mathscr{K}_\infty$. 
    With this, we define $\tilde{\alpha}_2 \in \mathscr{K}_\infty$ by
    \begin{equation*}
        \tilde{\alpha}_2(l) \defi \begin{cases*}
            \alpha_2^\prime(l), & $l \in [0,l_\mrm{min})$, \\
            \hat{\alpha}_2(l) + \bar{\alpha}_2(l-l_\mrm{min}), & $l \geq l_\mrm{min}$,
        \end{cases*}
    \end{equation*}
    which satisfies $\tilde{V}_N(\initstate,0) \leq \tilde{\alpha}_2(\| \initstate \|)$ for all $\initstate \in S$.
    
    Secondly, we note that $S \times \{ 0 \} \subseteq M_1$ holds.
    Hence, using \lemmaref{lemma:lyap_decrease}, for any initial value $p_\mrm{s} = (\initstate,0) \in S \times \{ 0 \}$ we obtain 
    $
        \tilde{V}_N(\mathpzc{p}(t,p_\mrm{s})) \geq \tilde{\ell}(\mathpzc{p}(t,p_\mrm{s})) \geq \lambda_\mrm{min}(Q) \| \mathpzc{x}(t,\initstate) \|^2 \rotdefi \tilde{\alpha}_1(\| \mathpzc{x}(t,\initstate) \|)
    $ and $\tilde{V}_N(\mathpzc{p}(t,p_\mrm{s})) \leq \tilde{V}_N(p_\mrm{s})$ for all $t \in \mathbb{N}_0$. 
    This implies
    \begin{multline*}
        \| \mathpzc{x}(t,\initstate) \| \leq \tilde{\alpha}_1^{-1}(\tilde{V}_N(\mathpzc{p}(t,p_\mrm{s})))
        \leq \tilde{\alpha}_1^{-1}(\tilde{V}_N(p_\mrm{s})) \\
        \leq (\tilde{\alpha}_1^{-1} \circ \tilde{\alpha}_2)(\| \initstate \|) \rotdefi \alpha_\mrm{S}(\| \initstate \| ) 
    \end{multline*}
    for all $\initstate \in S$ and $t \in \mathbb{N}_0$.
\end{proof}

\begin{remark}
    From a computational perspective, in \algref{alg:reducedRHC}, once $\mathpzc{x}(t,\initstate) \in \mathbb{X}_\mrm{f}$, it is possible to replace $\tilde{\kappa}$ with the terminal controller $\kappa_\mrm{f}$, in case $\kappa_\mrm{f}$ is easier to evaluate.
    This is, \eg the case when the terminal ingredients are based on the unconstrained LQR controller. 
    Then, $\kappa_\mrm{f}(x) = K_\infty x$, \cf \cite{Borrelli2017}.
\end{remark}

\section{DATA-DRIVEN SUBSPACE DESIGN}
\label{sec:data_driven_design}

As introduced in \secref{sec:reduced_mpc}, it is possible to design a reduced-order and stable MPC scheme based on a low-dimensional representation of the high-dimensional optimizers $\mu_N(x)$. 
A major assumption, however, is that the pair $(U,\sigma) \in \stiefel{r}{d} \times \set{F}_N$ in \algref{alg:reducedRHC} is initially admissible. 
This has to be incorporated into the design of $U$ and $\sigma$.
To this end, we state an important observation for linear systems, \cf \cite{Goebel2017}.  
\begin{lemma}
    \label{lemma:initial_feasibility}
	Let $(U,\sigma) \in \stiefel{r}{d} \times \set{F}_N$.
	Then:
 \begin{equation*}
    \resizebox{\columnwidth}{!}{$(U,\sigma)  \tn{ is IA} \Leftrightarrow
		\forall \bar{x}_j \in X_0 \colon \exists \bar{\alpha}_j \in \R^r \colon U \bar{\alpha}_j + \sigma(\bar{x}_j) \in \mathbb{U}^N(\bar{x}_j).$}
 \end{equation*}
\end{lemma}
\begin{proof}
    The $\Rightarrow$-direction is trivial, since by \assref{ass:initial_value}, $\bar{x}_j \in \mathbb{X}_0$ for all $j \in \{1,\dots,s\}$.
    For the $\Leftarrow$-direction, consider any $\initstate \in \mathbb{X}_0$.
    Since $\mathbb{X}_0 = \conv(X_0)$, there is $\theta_j \geq 0$ for $j \in \{1,\dots,s\}$ with $\sum_{j=1}^s \theta_j = 1$ such that $\initstate = \sum_{j=1}^s \theta_j \bar{x}_j$.
    It follows from convexity arguments that $U \tilde{\alpha} + \sigma(\initstate) \in \mathbb{U}^N(\initstate)$, with $\tilde{\alpha} = \sum_{j=1}^s \theta_j \bar{\alpha}_j$.
\end{proof}
\begin{remark}
	Checking if the pair $(U,\sigma)$ is initially admissible thus boils down to checking whether it offers admissible input sequences for a \emph{finite} number of states -- namely for the vertices of $\mathbb{X}_0$. 
\end{remark}

While initial admissibility is an assumption that \emph{must} be met in \algref{alg:reducedRHC}, we also seek to design $(U,\sigma)$ such that they are able to explain the most important patterns of the high-dimensional optimizers in $\set{D}$, so that the re-formulation from \eqref{eq:CFTOC} to \eqref{eq:CFTOC_reduced} is justified.

In the remainder of the paper, though, we focus on identifying the \emph{optimal basis} $U \in \stiefel{r}{d}$, \ie we fix the offset to a specific $\sigma_0 \in \set{F}_N$. 
A possible way to select the data $(\Gamma_0,\xi_0)$ that describes the offset $\sigma_0$ is
$\xi_0 = \sum_{i=1}^L \omega_i z_i$, with $\omega_i \geq 0$, $\sum_{i=1}^L \omega_i = 1$, and $\Gamma_0 = Z (X_\mrm{s})^\dag$,  
where $Z = \begin{bmatrix} z_1, \dots, z_L \end{bmatrix}$, $X_\mrm{s} = \begin{bmatrix} x_1 - \xi_0, \dots, x_L - \xi_0 \end{bmatrix}$ and $(\cdot)^\dag$ denotes the Moore-Penrose pseudoinverse \cite{Goebel2017}\footnote{In the light of \lemmaref{lemma:initial_feasibility}, the weights $\omega_i$ could be used to favor optimal input sequences associated to states close to the vertices of $\mathbb{X}_0$.}. 

As the offset is now fixed, we introduce -- based on the original data set \eqref{eq:systemData} -- a \emph{shifted} data set
$$
    \set{D}_{\mrm{s}} \defi \{
    z_i - \sigma_0(x_i) \mid (x_i,z_i) \in \set{D}
    \} ,
$$
and the \emph{shifted} admissible sets $\set{P}_j \defi \{ z - \sigma_0(\bar{x}_j) \mid z \in \mathbb{U}^N(\bar{x}_j) \}$ for $j \in \{1,\dots,s\}$.

We start our discussion by recalling an important fact concerning suboptimality and admissibility of \eqref{eq:CFTOC_reduced} from \cite[Lemma~1]{Schurig2023}. 
It states that whether $\widetilde{\mathscr{P}}_N(x_0,0)$ can be solved with the pair $(U,\sigma_0) \in \stiefel{r}{d} \times \set{F}_N$ solely depends on the \emph{subspace} spanned by the columns of $U$, and not directly on $U$ itself.
To see that, consider $\widetilde{U} = U R \in \stiefel{r}{d}$ with $R \in \og{r}$, where $\og{r}$ 
denotes the orthogonal group of dimension $r$.
Note that $\Span(\widetilde{U}) = \Span(U)$.
Solving $\widetilde{\mathscr{P}}_N(x,0)$ with either $(U,\sigma_0)$ or $(\widetilde{U},\sigma_0)$ leads to the same solution $\widetilde{V}_N(x,\tilde{z})$
, since the part $U \alpha^\star$ of the optimizer can also be represented by $\widetilde{U}$ via $\widetilde{U} \tilde{\alpha}^\star = \widetilde{U} (R^\top \alpha^\star) = U \alpha^\star$. 

This gives rise to the \emph{equivalence class}
$
	\eqclass{U} \defi \{ U R \in \stiefel{r}{d} \mid R \in \og{r} \} 
$,
and motivates the introduction of the so-called \emph{Grassmann manifold} \cite{Boumal2023, Edelman1998, Bendokat2020} as the quotient space
$
	\grass{r}{d} \defi \stiefel{r}{d} / \og{r} = \{ \eqclass{U} \mid U \in \stiefel{r}{d} \} 
$.

Intuitively, the Grassmann manifold $\grass{r}{d}$ is the set of all $r$-dimensional subspaces of $\R^d$. 
A subspace $\eqclass{U} \in \grass{r}{d}$ can be equivalently expressed by the unique orthogonal projector $P = \pi(U) \defi U U^\top$ onto it \cite{Bendokat2020}; see \cite{Schurig2023} for a more elaborate discussion.

\begin{figure}
    \centering
    \begin{tikzpicture}[line width=1.25pt,>=stealth]
    \draw[->] (0,0) -- node[pos=1, below] {\footnotesize $z^1$} +(2.5, 0);
    \draw[->] (0,0) -- node[pos=1, left] {\footnotesize $z^2$} +(0, 2.5);
    
    \node at (0.75,2.25)  (z) [circle, draw, inner sep=0.75, fill=white] {};
    \node at (z) [left,CCPSgray1] {\footnotesize{$\delta_i$}};
    
    \draw[-,CCPSredBound] (0,0) -- node[pos=1,right] {\footnotesize${\Span(U)}$} (2.5,2.5);
    \node at (1.5, 1.5)  (zp) [circle, draw, inner sep=0.75, fill=white] {};
    \node at ([xshift=0.75pt]zp.east) [below,CCPSredBound] {\footnotesize{$\delta_{\mrm{p},i}$}};
    
    \draw[->,CCPSredBound,dashed] (z) -- node[right,pos=0.4] {\footnotesize $P \delta_i$} ([xshift=-0.5pt,yshift=0.5pt]zp.north west);
    
    \draw[->] (3.75,0.75) -- ++(2.25,0) node [below] {$\R$};
    \draw[CCPSturquoise1] (5.5,0.75) node[below] {\footnotesize $\alpha_i$};
    \draw[] (5.5,0.75) node[name=alpha] {\tiny $|$};
    \draw[CCPSblue1] (4,0.75) node[below] {\footnotesize $\widetilde{\alpha}_i$};
    \draw[] (4,0.75) node[name=alphatilde] {\tiny $|$};
    
    \draw[->,CCPSturquoise1] ([xshift=2pt]zp.east) -- node[above] () {\footnotesize $U^\top \delta_i$} (5.45,0.8);
    \draw[->,CCPSblue1] ([xshift=2pt]zp.east) -- node[below] () {\footnotesize $\widetilde{U}^\top \delta_i$} (3.95,0.8);
\end{tikzpicture}
    \caption{Illustration of the subspace approach for $d=2$, $r=1$.} \vspace{-3ex}
    \label{fig:projection}
\end{figure}
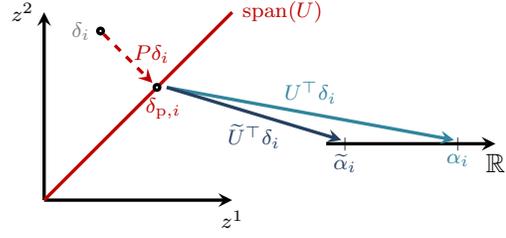

The upcoming discussion is supported by \figref{fig:projection}.
We consider $\delta_i \in \set{D}_\mrm{s}$, and an arbitrary subspace $\eqclass{U} \cong P = \pi(U)$.
Firstly, we can \emph{project} $\delta_i$ onto $\Span(U)$ by virtue of the projector $P$.
As we can see, $\delta_{\mrm{p},i} \defi P \delta_i$ is the point in $\Span(U)$ with the minimal Euclidean distance to $\delta_i$.
The point $\delta_{\mrm{p},i} \in \Span(U)$ also has a {low-dimensional representation} $\alpha_i \defi U^\top \delta_i \in \R^r$.
We refer to $\alpha_i$ as the \emph{latent variable} for $\delta_i$ with respect to $U$.
Note that $U \alpha_i = \delta_{\mrm{p},i}$.

An important observation, though, is that the latent variable depends on the \emph{basis} that is chosen to span the subspace, and not on the subspace itself.
Consider $\widetilde{U} = U R \in \eqclass{U}$ with $\widetilde{U} \neq U$.
Surely, there exists $\tilde{\alpha}_i \in \R^r$ such that $\widetilde{U} \tilde{\alpha}_i = \delta_{\mrm{p},i}$, but we have the relation $\tilde{\alpha}_i = R^\top \alpha_i \neq \alpha_i$.

Because the feasibility of $\widetilde{\mathscr{P}}_N(x,0)$ solely depends on the subspace, the latent variable is not well suited to represent a point in $\Span(U)$ in the design phase of the subspace since it loses its meaning and interpretation under a change of basis.
Hence, we suggest the following design task. 
\begin{problem}
    Design a subspace $\eqclass{U} \in \grass{r}{d}$ such that
    \begin{enumerate}
        \item[(i)] The subspace $\eqclass{U}$ intersects all shifted admissible sets: 
        $
            \forall j \in \{1,\dots,s\} \colon \exists \tilde{z}_j \in \Span(U) \colon \tilde{z}_j \in \set{P}_j 
        $.
        \item[(ii)] The subspace $\eqclass{U}$ minimizes the Euclidean distance of $\delta_i$ to $\Span(U)$ for all $\delta_i \in \set{D}_\mrm{s}$.
    \end{enumerate}
    \label{prob:subspace_design}
\end{problem}
\pagebreak
\begin{remark}
    More elaborate optimality requirements than (ii) are possible on $\grass{r}{d}$. We restrict our analysis to this simple one for the sake of an easy interpretation and exposition.
\end{remark}

\subsection{Design of the optimal subspace}

We present an optimization problem over $\grass{r}{d}$ as a possible solution to \probref{prob:subspace_design}, and compare our approach to an existing one that is based on convex programming.

\subsubsection{Riemannian approach}
We aim to minimize the distance between the points in $\set{D}_\mrm{s}$ and the subspace $\eqclass{U} \in \grass{r}{d}$. 
Following our previous discussion, we directly consider the {projection} $P \delta_i$ of $\delta_i$ onto $\Span(U)$ via the projector $P = \pi(U)$. 
This retains the interpretability of the formulation under a change of basis $\widetilde{U} \in \eqclass{U}$.
Thus, we introduce 
$
    f \colon \grass{r}{d} \to \R_{\geq 0}, \;
    P \mapsto \sum_{i=1}^L \| \delta_i - P \delta_i \|^2
$
as objective function, \cf \figref{fig:projection}. 

We can now turn to the initial admissibility constraint, which is illustrated in \figref{fig:demo_low_dim_rep}.
We show the shifted admissible set $\set{P}_j$ for one $j \in \{1, \dots, s\}$ and a subspace $\eqclass{U} \cong P= \pi(U) \in \grass{r}{d}$. 
We again rely on the concept of {projection}. 
For admissibility, we consider the maximum volume ellipsoid center $\bar{\delta}_j$ of $\set{P}_j$ \cite[Sec.~8.5.2]{Boyd2004}; this center can also be seen in \figref{fig:demo_low_dim_rep}. 
The idea is to demand that the point on $\Span(U)$ with the smallest Euclidean distance to the center $\bar{\delta}_j$ of $\set{P}_j$ should still be in $\set{P}_j$. 
As discussed before, $P \bar{\delta}_j \in \Span(U)$ is exactly this point, so we require $P \bar{\delta}_j \in \set{P}_j$ for all $j \in \{1,\dots,s\}$. 
In conclusion, to tackle \probref{prob:subspace_design}, we propose to solve
\begin{align}
    &\min_{P \in \grass{r}{d}} \, f(P) &
    &\mrm{s.t.} & 
    &\forall j \in \{1,\dots,s\} \colon P \bar{\delta}_j \in \set{P}_j.
    \label{eq:riemannian_formulation} 
\end{align}
\begin{corollary}
	Assume that $P = \pi(U) \in \grass{r}{d}$ is a feasible point of \eqref{eq:riemannian_formulation}. 
	Then, the pair $(U,\sigma_0) \in \stiefel{r}{d} \times \set{F}_N$ is initially admissible.
\end{corollary}
\begin{proof}
    For each $j \in \{1,\dots,s\}$, we have $U \bar{\alpha}_j \in \set{P}_j$ with $\bar{\alpha}_j = {U}^\top \bar{\delta}_j$ by assumption.
    Then, $U \bar{\alpha}_j + \sigma_0(\bar{x}_j) \in \mathbb{U}^N(\bar{x}_j)$ by definition of $\set{P}_j$.
\end{proof}

\subsubsection{Comparison to a Euclidean approach}

In \cite{Goebel2017}, the authors proposed an approach similar to \eqref{eq:riemannian_formulation}. 
We refer to the approach in \cite{Goebel2017} as an \qq{Euclidean approach}, since the rigorous viewpoint of optimizing over subspaces was not pursued. 
We present a slightly adapted version of the Euclidean approach for comparison with our Riemannian approach, to highlight its conceptual advantages. 

In the Euclidean approach, the idea is to \emph{alternately} optimize the basis and the latent representation of points in $\set{D}_\mrm{s}$ and $\set{P}_j$.
Starting from an initial subspace $\eqclass{U_{(0)}} \in \grass{r}{d}$, we construct $\eqclass{U_{(k+1)}} \in \grass{r}{d}$ from $\eqclass{U_{(k)}}$ in an iterative fashion. 
Once a subspace $[U_{(k)}]$ is available, we seek to find a latent representation $\bar{\alpha}_j^{(k)}$ for a point on $\Span(U_{(k)})$ which is also in $\set{P}_j$, or at least close to it. 
Following our previous discussion, we again choose the point on $\Span(U_{(k)})$ with the smallest Euclidean distance to the maximum volume ellipsoid center $\bar{\delta}_j$ of $\set{P}_j$.
This choice leads to the latent variables $\bar{\alpha}_j^{(k)} \defi U_{(k)}^\top \bar{\delta}_j$. 
The convex program
\begin{align}
    \hat{U}_{(k+1)} = &\arg \min_{U \in \R^{d \times r}} \, f^\mrm{eucl}(U) &
    &\mrm{s.t.} & 
    &U \bar{\alpha}_j^{(k)} \in \set{P}_j
    \label{eq:euclidean_formulation} 
\end{align}
is solved, with $j \in \{1,\dots,s\}$. 
Then, a subspace $\eqclass{U_{(k+1)}} \in \grass{r}{d}$ can be obtained by orthogonalizing the columns of $\hat{U}_{(k+1)}$, the latent variables $\bar{\alpha}_j^{(k+1)}$ can be computed as described above and the iteration counter $k$ can be increased. 
Here, $f^\mrm{eucl}$ stands for any convex cost function; we concentrate on the feasibility of \eqref{eq:euclidean_formulation}.

The above procedure enjoys the benefit of entirely relying on the (repetitive) solution of a convex program, for which powerful tools exist. 
However, we argue that our subspace perspective offers further potential \emph{conceptually}. 
The key difference to \eqref{eq:riemannian_formulation} is that \eqref{eq:euclidean_formulation} relies on the latent representation $\bar{\alpha}_j^{(k)}$ of $\bar{\delta}_j$ \emph{with respect to $U_{(k)}$}. 
We have already discussed that the latent representation loses its meaning under a change of basis. 
As we see next in a simple example, it is also not suited for comparing high-dimensional points among different subspaces.

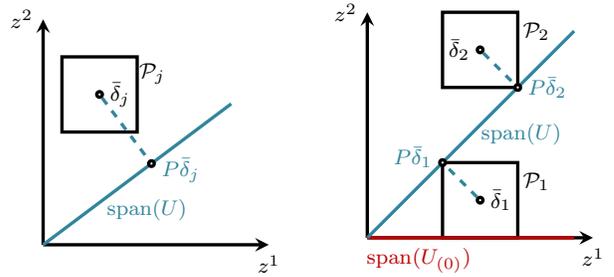
\begin{figure} 
    \centering
    
    \subfloat[Riemannian approach.]{
        \begin{tikzpicture}[line width=1.25pt,>=stealth,scale=0.5]
    \draw[->] (0,0) -- node[pos=1, below] {\footnotesize $z^1$} +(6, 0);
    \draw[->] (0,0) -- node[pos=1, left] {\footnotesize $z^2$} +(0, 6);
    \draw [] (0.5,3) rectangle ++(2,2);
    \node at (2.9,4) [above] {\footnotesize $\set{P}_j$};
    \node at (1.5,4)  (zc) [circle, draw, inner sep=0.75, fill=white] {};
    \node at (zc) [right] {\footnotesize{$\bar{\delta}_j$}};
    
    \draw[-,CCPSturquoise1] (0,0) -- node[pos=0.25,right,xshift=2.5pt] {\footnotesize${\Span(U)}$} (5,3.75);
    \node at (2.88, 2.16)  (z0) [circle, draw, inner sep=0.75, fill=white] {};
    \node at (z0) [right,yshift=-2.5pt,CCPSturquoise1] {\footnotesize{$P \bar{\delta}_j$}};
    
    \draw[-,CCPSturquoise1,dashed] (zc) -- (z0);
\end{tikzpicture}
        \label{fig:demo_low_dim_rep}
    }
    \subfloat[Infeasibility of the Euclidean approach.]{
        \begin{tikzpicture}[line width=1.25pt,>=stealth,scale=0.5]
    \draw[->] (0,0) -- node[pos=1, below] {\footnotesize $z^1$} +(6, 0);
    \draw[->] (0,0) -- node[pos=1, left] {\footnotesize $z^2$} +(0, 6);
    
    \draw [] (2,0) rectangle ++(2,2);
    \node at (4.5,1.5) {\footnotesize $\set{P}_1$};
    \node at (3,1)  (zc1) [circle, draw, inner sep=0.75, fill=white] {};
    \node at (zc1) [right] {\footnotesize{$\bar{\delta}_1$}};

    \draw [] (2,4) rectangle ++(2,2);
    \node at (4.5,5.5) {\footnotesize $\set{P}_2$};
    \node at (3,5)  (zc2) [circle, draw, inner sep=0.75, fill=white] {};
    \node at (zc2) [left] {\footnotesize{$\bar{\delta}_2$}};
    
    \draw[-,CCPSredBound] (0,0) -- node[pos=0.25,below] {\footnotesize${\Span(U_{(0)})}$} (5.5,0);
    \draw[-,CCPSturquoise1] (0,0) -- node[pos=0.5,right] {\footnotesize${\Span(U)}$} (5.5,5.5);
    
    \node at (2, 2)  (z1) [circle, draw, inner sep=0.75, fill=white] {};
    \node at (z1) [left,CCPSturquoise1,yshift=2pt] {\footnotesize{$P \bar{\delta}_1$}};
    \node at (4, 4)  (z2) [circle, draw, inner sep=0.75, fill=white] {};
    \node at (z2) [right,CCPSturquoise1] {\footnotesize{$P \bar{\delta}_2$}};
    
    \draw[-,CCPSturquoise1,dashed] (zc1) -- (z1);
    \draw[-,CCPSturquoise1,dashed] (zc2) -- (z2);
\end{tikzpicture}
        \label{fig:demo_low_dim_rep_eucl}
    }
    \caption{Illustration of the admissibility problem.} \vspace{-3ex}
    \label{fig:demo_ad}
\end{figure}

\begin{example}
    We assume $d = s = 2$, and consider the polyhedral sets $\set{P}_1 = \{ z \in \R^2 \mid (2,0) \leq z \leq (4,2) \}$ and $\set{P}_2 = \{ z \in \R^2 \mid (2,4) \leq z \leq (4,6) \}$. 
    They have the maximum volume ellipsoidal centers $\bar{\delta}_1 = (3,1)$ and $\bar{\delta}_2 = (3,5)$, respectively. 
    This is illustrated in \figref{fig:demo_low_dim_rep_eucl}.
    We consider an initial subspace spanned by $U_{(0)}^\top = \begin{bmatrix} 1 & 0 \end{bmatrix}$, which implies $\bar{\alpha}_1^{(0)} = \bar{\alpha}_2^{(0)} = 3$. 
    The optimization variable in \eqref{eq:euclidean_formulation} is $U^\top = \begin{bmatrix} u_1 & u_2 \end{bmatrix}$, and for our specific initial subspace the constraints for $u_2$ are given by $0 \leq 3 u_2 \leq 2 \, \land \, 4 \leq 3 u_2 \leq 6$, \ie \eqref{eq:euclidean_formulation} is infeasible.
    Note that the subspace spanned by $U^\top = \frac{1}{\sqrt{2}} \begin{bmatrix} 1 & 1 \end{bmatrix}$ does intersect both sets, \cf \figref{fig:demo_low_dim_rep_eucl}, so that $\eqclass{U}$ is a possible solution to \probref{prob:subspace_design}. 
    This fact is just not captured by the Euclidean formulation \eqref{eq:euclidean_formulation}; it is, however, captured by the Riemannian formulation \eqref{eq:riemannian_formulation}. 
\end{example}
While this example is admittedly constructed, it shows the major flaw of the Euclidean approach: the latent variables $\bar{\alpha}_j^{(k)}$ \emph{only} have a meaning in the coordinates given by $U_{(k)}$, and not necessarily in the coordinates given by a different~$U$.

\begin{remark}
    While \eqref{eq:riemannian_formulation} enjoys the aforementioned conceptional advantages, it comes at the price of having to optimize over the Riemannian manifold $\grass{r}{d}$.
    This can be dealt with by using tailored algorithms \cite{Absil2008, Boumal2023}; see also \cite{Schurig2023}.
    
    Further, we consider \emph{additional} constraints of the form $P \bar{\delta}_j \in \set{P}_j$, which can be tackled by applying methods from \cite{Liu2020}.
    A detailed discussion, though, is out of scope here.
\end{remark}

\section{SIMULATION RESULTS}
\label{sec:example}

We consider the system $(\dot{x}^1, \dot{x}^2) = (x^2,x^1+u)$, which is obtained by linearizing a simple inverted pendulum model around its upright position. 
Here, $x^1$ and $x^2$ describe the angular  position and velocity, and $u$ the applied torque. 
We discretize with sampling time $T_\mrm{s} = \valunit{0.1}{s}$ to obtain a discrete-time model.
We impose the state and input constraints $\mathbb{X} = \{ x \in \R^2 \mid | x^1 | \leq 1, \, | x^2 | \leq 0.35 \}$ and $\mathbb{U} = \{ u \in \R \mid | u | \leq 1 \}$.
The objective is to stabilize the linearized system in the origin.
Hence, we apply an MPC scheme with stage cost 
$\ell(x,u) = x^\top x + 0.1 u^2$. 
The unconstrained LQR controller $K_\infty$ is used to derive the terminal cost and terminal set.
The MPT3 toolbox and YALMIP are used to formulate \eqref{eq:CFTOC}.

Using \algref{alg:reducedRHC}, we seek to design a reduced MPC scheme that features a subspace of dimension $r=2$.
As our baseline MPC, we consider $\mathscr{P}_{13}(x)$ in \eqref{eq:CFTOC} for the prediction horizon $N = 13$. 
We use $K_\infty$ to pre-stabilize the system.

We aim to achieve the same domain of attraction that the full-order MPC controller would have for a prediction horizon of $N_\mrm{desired} = 12$, \ie we set the initial set to $\mathbb{X}_0 = \set{X}_{12}$.
We obtain the data set $\set{D}$ in \eqref{eq:systemData} by sampling $L=450$ states $x_i \in \mathbb{X}_0 \setminus \mathbb{X}_\mrm{f}$, and computing the corresponding optimal input sequences $z_i = \mu_{13}(x_i) \in \mathbb{U}^{13}(x_i)$.
 \figref{fig:simple_pendulum_setup} illustrates the situation.

\begin{figure}
    \centering
    \input{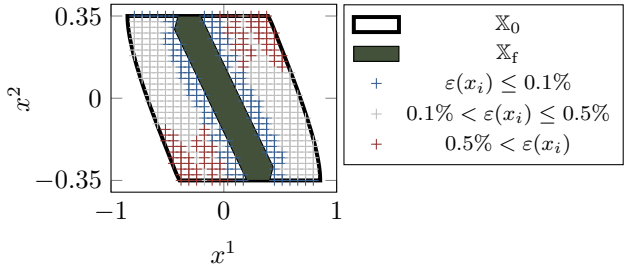}
    \caption{Initial set and data points for pendulum example.} \vspace{-3ex}  
    \label{fig:simple_pendulum_setup}
\end{figure}

We determine the fixed offset $\sigma_0 \in \set{F}_{13}$ as described in \secref{sec:data_driven_design} with $\omega_i = 1/L$, and apply our Riemannian framework \eqref{eq:riemannian_formulation} to design a subspace $\eqclass{U^\star} \in \grass{2}{13}$ such that the pair $(U^\star,\sigma_0)$ is initially admissible.
For this task, we implement \cite[Alg.~1]{Liu2020}, and the optimization problem(s) on the Grassmann manifold are solved with Manopt \cite{Boumal2014b}.
We note that the {Euclidean} counterpart in \eqref{eq:euclidean_formulation} was infeasible in the first iteration $k=0$.

In \figref{fig:simple_pendulum_setup}, we report the relative difference $\varepsilon \colon \mathbb{X}_0 \to \R$ for the closed-loop cost $\tilde{J}_\mrm{cl}$ with respect to the full-order MPC controller with $N=12$.
The mean of $\varepsilon$ over the shown grid is $0.31 \%$, with a standard deviation of $0.34 \%$.

To conclude, with a two-dimensional subspace, we designed a reduced MPC scheme which uses three degrees of freedom in closed-loop to guarantee stability and feasibility. 
The scheme has the same domain of attraction as the full-order controller with a prediction horizon of twelve, and only a minor decrease in closed-loop performance.

\section{CONCLUSIONS}
\label{sec:conclusion}

We introduce a reduced-order MPC scheme for which stability and feasibility can be established. 
The required subspace is computed via constrained optimization over a Riemannian manifold.  
The design process ensures initial feasibility of the MPC problem in the reduced space, which paves the way for the closed-loop properties.

There are many possible next developments. 
In this work, we have fixed all design parameters but the subspace with a heuristic, so approaches can be discussed that take all parameters -- such as the offset or the pre-stabilizing feedback gain -- and their interconnections into account.
Moreover, the current design process does not consider closed-loop information. 
Reinforcement learning or Bayesian optimization are possible tools for tackling the outlined challenges.









\bibliographystyle{ieeetr}
\bibliography{./bib/literature}

\begin{thebibliography}{10}

\bibitem{Findeisen2002}
R.~Findeisen and F.~Allg{\"o}wer, ``{An Introduction to Nonlinear Model
  Predictive Control},'' in {\em 21st Benelux Meeting on Systems and Control},
  pp.~119--141, 2002.

\bibitem{Gruene2017}
L.~Gr{\"u}ne and J.~Pannek, {\em Nonlinear Model Predictive Control : Theory
  and Algorithms.}
\newblock Communications and Control Engineering, Springer, 2.~ed., 2017.

\bibitem{Borrelli2017}
F.~Borrelli, A.~Bemporad, and M.~Morari, {\em Predictive Control for Linear and
  Hybrid Systems}.
\newblock Cambridge University Press, 2017.

\bibitem{Cagienard2007}
R.~Cagienard, P.~Grieder, E.~Kerrigan, and M.~Morari, ``Move blocking
  strategies in receding horizon control,'' {\em J. Proc. Cont.}, vol.~17,
  no.~6, pp.~563--570, 2007.

\bibitem{Shekhar2015}
R.~C. Shekhar and C.~Manzie, ``Optimal move blocking strategies for model
  predictive control,'' {\em Automatica}, vol.~61, pp.~27--34, 2015.

\bibitem{Pan2023}
G.~Pan and T.~Faulwasser, ``{NMPC} in active subspaces: Dimensionality
  reduction with recursive feasibility guarantees,'' {\em Automatica},
  vol.~147, p.~110708, 2023.

\bibitem{Bemporad2020}
A.~Bemporad and G.~Cimini, ``Variable elimination in model predictive control
  based on {K-SVD} and {QR} factorization,'' {\em IEEE Transactions on
  Automatic Control}, vol.~68, no.~2, pp.~782--797, 2021.

\bibitem{Rojas2004}
O.~J. Rojas, G.~C. Goodwin, M.~M. Serón, and A.~Feuer, ``An {SVD} based
  strategy for receding horizon control of input constrained linear systems,''
  {\em Int. J. Rob. Cont.}, vol.~14, no.~13-14, pp.~1207--1226, 2004.

\bibitem{Goebel2017}
G.~Goebel and F.~Allgöwer, ``Semi-explicit mpc based on subspace clustering,''
  {\em Automatica}, vol.~83, pp.~309--316, 2017.

\bibitem{Schurig2023}
R.~Schurig, A.~Himmel, and R.~Findeisen, ``Toward grassmannian dimensionality
  reduction in mpc,'' {\em IEEE Control Systems Letters}, vol.~7,
  pp.~3187--3192, 2023.

\bibitem{Absil2008}
P.-A. Absil, R.~Mahony, and R.~Sepulchre, {\em Optimization Algorithms on
  Matrix Manifolds}.
\newblock Princeton, NJ: Princeton University Press, 2008.

\bibitem{Boumal2023}
N.~Boumal, {\em An introduction to optimization on smooth manifolds}.
\newblock Cambridge University Press, 2023.

\bibitem{Liu2020}
C.~Liu and N.~Boumal, ``Simple algorithms for optimization on riemannian
  manifolds with constraints,'' {\em Applied Mathematics \& Optimization},
  vol.~82, pp.~949--981, 2020.

\bibitem{Brunton2019}
S.~L. Brunton and J.~N. Kutz, {\em Data-Driven Science and Engineering: Machine
  Learning, Dynamical Systems, and Control}.
\newblock Cambridge University Press, 2019.

\bibitem{Edelman1998}
A.~Edelman, T.~A. Arias, and S.~T. Smith, ``The geometry of algorithms with
  orthogonality constraints,'' {\em SIAM J. Matrix Anal.\& Appl.}, vol.~20,
  no.~2, pp.~303--353, 1998.

\bibitem{Bendokat2020}
T.~Bendokat, R.~Zimmermann, and P.-A. Absil, ``A {Grassmann} manifold handbook:
  Basic geometry and computational aspects,'' {\em arXiv preprint
  arXiv:2011.13699}, 2020.

\bibitem{Boyd2004}
S.~P. Boyd and L.~Vandenberghe, {\em Convex optimization}.
\newblock Cambridge university press, 2004.

\bibitem{Boumal2014b}
N.~Boumal, B.~Mishra, P.-A. Absil, and R.~Sepulchre, ``{M}anopt, a {M}atlab
  toolbox for optimization on manifolds,'' {\em Journal of Machine Learning
  Research}, vol.~15, no.~42, pp.~1455--1459, 2014.

\end{thebibliography}

\end{document}